\documentclass[a4paper,11pt]{article}
\usepackage{amssymb,amsmath,amsthm}
\usepackage[dvips]{graphicx}
\usepackage{makeidx}
\usepackage{color,calc}
\usepackage{dsfont}
\usepackage{hyperref}
\usepackage{array}
\usepackage{mathtools}
\usepackage[normalem]{ulem}

\newtheorem{thm}{Theorem}[section]
\newtheorem{prop}[thm]{Proposition}

\newtheorem{defin}[thm]{Definition}

\theoremstyle{remark}

\newtheorem{ex}[thm]{Example}

\hypersetup{
    colorlinks=true,
    linkcolor=blue,
    filecolor=magenta,      
    urlcolor=cyan,
} 
\makeindex

\begin{document}

\title{Tetrahedron maps and symmetries of three dimensional integrable discrete equations}


\author{P.~Kassotakis \and M.~Nieszporski \and V. Papageorgiou \and   A.~Tongas}%

\newcommand{\Addresses}{{
  \bigskip
  \footnotesize

  P.~Kassotakis, \textsc{Department of Mathematics and Statistics, 
                         University of Cyprus, P.O Box: 20537, 
                         1678 Nicosia, Cyprus}
  \par\nopagebreak
\textit{E-mail address}, P.~Kassotakis: 
\texttt{pavlos1978@gmail.com, pkasso01@ucy.ac.cy}

  \medskip

  M.~Nieszporski, \textsc{Katedra  Metod  Matematycznych  Fizyki,  
                          Wydzia l  Fizyki,  Uniwersytet  Warszawski,  
                          ul.   Pasteura  5,02-093 Warszawa, Poland}
 \par\nopagebreak
\textit{E-mail address}, M.~Nieszporski: \texttt{maciejun@fuw.edu.pl}

  \medskip

  V.~Papageorgiou, \textsc{Department of Mathematics, University of Patras, 
                           26 504 Patras, Greece}
  \par\nopagebreak
  \textit{E-mail address}, V.~Papageorgiou: \texttt{vassilis@math.upatras.gr}

  \medskip
  
  A.~Tongas, \textsc{Department of Mathematics, University of Patras, 
                     26 504 Patras, Greece}
  \par\nopagebreak
  \textit{E-mail address}, A.~Tongas: \texttt{tasos@math.upatras.gr}

}}

\maketitle

\begin{abstract}

A relationship between the tetrahedron equation for maps and the consistency property 
of integrable discrete equations on $\mathbb{Z}^3$ is investigated. 
Our approach is a generalization of a method developed in the context of Yang--Baxter 
maps, based on the invariants of symmetry groups of the lattice equations. 
The method is demonstrated by a case--by--case analysis of the octahedron type lattice 
equations classified recently, leading to some new examples 
of tetrahedron maps and integrable coupled lattice equations.

\end{abstract}

\tableofcontents

\section{Introduction}

The purpose of this work is to explain that the study of discrete integrable equations on 
$\mathbb{Z}^3$ and their symmetry group of transformations is intimately connected with 
solutions of the functional tetrahedron (Zamolodchikov) equation \cite{zamo_1980}, 
or simply tetrahedron maps. 
This idea is not new and similar considerations have been applied to integrable lattice 
equations on $\mathbb{Z}^2$ \cite{abs2003}, leading to a systematic derivation 
of set theoretic solutions of the Yang--Baxter equation  
\cite{sklyanin-1988,drinfeld-1992,etingof-1999,veselov2003},
or simply Yang--Baxter maps, see \cite{veselov2003,ptv2006,vastas2007,pstv2010,sokor_2013,gra_2016} 
on this interplay.
Therefore it is reasonable to expect that such a correspondence exists also in three 
dimensions and our aim, in this work, is to give a detailed analysis of this link  
and the implications involved. 

Two classification results will play a significant role to our considerations. 
The first was given by Adler, Bobenko and Suris in \cite{absocta}, concerning a classification 
of integrable discrete equations on $\mathbb{Z}^3$ based on the consistency property, 
nowadays a synonymous of integrability. All classified lattice equations, even in non--commutative 
version, have appeared earlier in the literature, see \cite{nc1990}. 
The second was obtained by Sergeev in \cite{sergeev1998}, see also \cite{kks98,kas1996}, 
where a classification of tetrahedron maps was presented, based on the local Yang--Baxter
equations \cite{Maillet:1989,Maillet:1989b}, which serve as a zero--curvature condition 
of the corresponding maps. 
We will show that the invariants of the symmetry groups of transformations of the integrable
lattice equations on $\mathbb{Z}^3$, become the variables to express the tetrahedron maps 
from the latter list.

However, the tetrahedron maps obtained here are not exhausted to the classification results in 
\cite{sergeev1998}. In fact, our method leads to vector extensions of some of tetrahedron maps 
appearing in \cite{sergeev1998}, which may be viewed as Bianchi permutability of B\"acklund 
transformations for coupled PDEs of KP type. Furthermore, lattice equations give rise to a wider 
class of tetrahedron maps. This is a consequence of the multi--parameter symmetry group admitted 
by the lattice equations and their classification into non--equivalent two dimensional subgroups. 
Remarkably, imposing multiple consistent copies of the same lattice equation on all cubes of a 
hypercube, and considering the joint invariants of the full symmetry group of the equation, 
we again derive other tetrahedron maps, reflecting the multi--dimensional consistency property 
of the lattice equation.

The method may be exploited in the opposite order. Given a tetrahedron map one may find an 
integrable discrete system associated with it, and this correspondence 
is demonstrated by certain examples leading to non--commutative versions of lattice KP 
type equations. The associated discrete systems for Yang--Baxter maps were investigated in 
\cite{KaNie,KaNie1,KaNie3,KaNie:2018,Kouloukas:2012}

The structure of the paper is as follows. Section \ref{symm} contains the necessary background.
First we give an overview of the functional tetrahedron equation and its symmetries.
Next, we give the basics on symmetry groups of transformations of lattice equations. 
Particular emphasis is given on the dual formulation of Frobenius Theorem in terms of differential 
forms, which provides us an elegant way for constructing a complete set of $G$--invariants 
under a regular action. The motivating examples here are AKP and BKP, which share the same 
local symmetry group.

In Section \ref{class} a case--by--case analysis is developed for the integrable 3D lattice 
equations of octahedron type appeared in the classification \cite{absocta} by
Adler, Bobenko and Suris, as well as for BKP equation. 
The most prominent example is discrete potential KP equation (denoted by ($\chi_3$) in
\cite{absocta}) and its symmetry group $G \cong {\rm{Aff}}(\mathbb{C})$, supplying three 
non--equivalent two dimensional subgroups, and the corresponding tetrahedron maps. 
Lattice equation ($\chi_5$) gives a map which obeys an entwining relation with the
tetrahedron map obtained from lattice equation ($\chi_4$). 

In Section \ref{multi} our motivating guide is the multi--dimensional consistency property 
of potential forms of discrete KdV and KP equations. We exploit this interplay to show
that a different tetrahedron map naturally arises by imposing $\chi_3$ on $\mathbb{Z}^4$. 
This fact is a reminiscence of the most generic quadrirational YB map obtained from discrete 
potential KdV using its full symmetry group on $\mathbb{Z}^3$, cf. \cite{ptv2006}.
Finally, the method is applied on the opposite order connecting some non--commutative versions
of tetrahedron maps with their corresponding lattice equations.

The final Section \ref{concl} of the paper consists of Conclusions and Perspectives, 
where an overall evaluation of the results obtained in the main body of the paper is presented, 
along with the description of various venues for expanding the above results.

\section{Symmetries of tetrahedron maps and lattice equations} \label{symm}
\begin{figure}[h!]
\centering
\def\svgscale{0.5}
  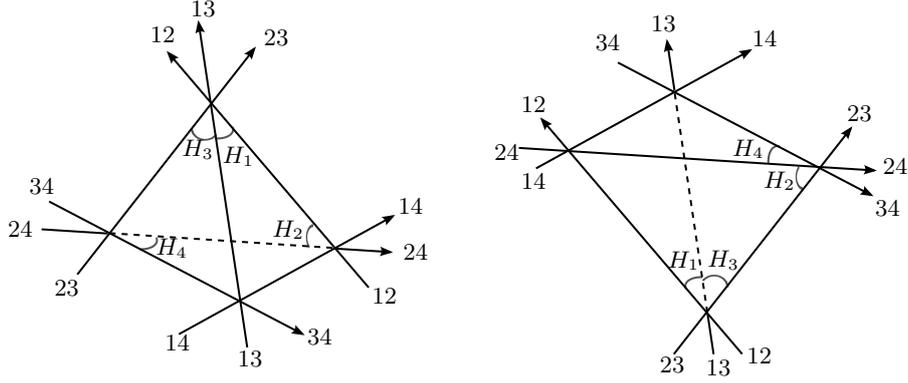
  \caption{Geometric interpretation of the tetrahedron relation.}
  \label{fig:tetrahedron}
\end{figure}
A geometric interpretation of the tetrahedron equation, or Zamolodchikov
equation, is given in Fig. \ref{fig:tetrahedron}, see \cite{zamo_1981}, \cite{kv}. 
In general position, a tetrahedron consists of four planes $H_i$, $i=1,2,3,4$ in 
$\mathbb{R}^3$, which are meeting along the lines $\ell_{ij}=H_i \cap H_j$, pairwise. 
Similarly, three planes meet on the four vertices $v_i$, $i=1,2,3,4$, e.g.  $v_1=H_2\cap H_3
\cap H_4$ etc. Lines are directed from vertex $v_i$ to $v_j$ with $i<j$. Moving one plane 
across the intersection point of the three other planes, we obtain a configuration of 
planes, lines and vertices depicted on the right of Fig. \ref{fig:tetrahedron}, 
regardless of which plane is moving. Let
$\widehat{i}$, denotes all pairs $(k,j)$ not containing $i$
$$\widehat{1} = \big(\, (23), (24), (34)\,\big)\,,\quad  \widehat{2} = \big(\, (13), (14),
(34)\,\big) \,,\quad \widehat{3} = \big(\, (12), (14), (24)\, \big)\,,\quad
\widehat{4} = \big(\, (12), (13), (23)\,\big) \,,$$ 
in lexicographic order, meaning that $(ij)$ precedes $(mn)$ if $i<m$ or $i=k$, and $j<n$.
Then tetrahedron equation may be regarded as the condition that the two
configurations of Fig \ref{fig:tetrahedron} lead to the same result
\begin{equation} \label{zamo1}
  R^{\widehat{4}} \, R^{\widehat{3}} \, R^{\widehat{2}} \, R^{\widehat{1}}\, =
  R^{\widehat{1}} \, R^{\widehat{2}} \, R^{\widehat{3}} \, R^{\widehat{4}} \,,
\end{equation}
Denoting the pairs $(ij)$ by
   $$1=(12)\, \quad 2=(13)\,, \quad 3=(23)\,,\quad 4=(14)\,,\quad
   5=(24)\,,\quad 6=(34)\,,$$
tetrahedron equation \eqref{zamo1} takes the form 
\begin{equation} \label{zamo2}
R^{(123)} \, R^{(145)} \, R^{(246)} \, R^{(356)} = R^{(356)} \, R^{(246)}
\, R^{(145)} \, R^{(123)} \,,
\end{equation}
which is the main subject of this paper.  

In its original form, tetrahedron equation is regarded as an equality of operators on 
the 6--tuple tensor product $V^{\otimes \, 6}$, of a vector space $V$. Here, 
$R^{(123)} \in \rm{End} (V^{\otimes \, 6})$ acts only on the first, second, and third
factor and as identity on the others, and similarly for $R^{(145)}$ etc. 
Replacing $V$ by any set $\mathbb{X}$ and regarding equation
\eqref{zamo2} as an equality of composition of maps, we arrive naturally to
the notion of {\em functional tetrahedron maps}.

\subsection{Symmetries of the tetrahedron equation}

\begin{defin} 
 Assume there is a bijection $\phi : \mathbb{X}\to\mathbb{X}$, 
 and $R  : \mathbb{X}^3\to\mathbb{X}^3$ satisfies the
 tetrahedron relation, 
 \begin{equation} \label{zamo3}
R^{(123)} \, R^{(145)} \, R^{(246)} \, R^{(356)} = R^{(356)} \, R^{(246)}
\, R^{(145)} \, R^{(123)} \,,
\end{equation}
regarded as an equality of composition of maps in $\mathbb{X}^6$.
Then the same is true for
\begin{equation}
\widetilde{R} =  \phi^{-1} \times \phi^{-1} \times \phi^{-1}\, R  \, \phi 
\times \phi \times \phi \,,
\end{equation}
and the maps $R$, $\widetilde{R}$, are called gauge equivalent, and
$\phi$, a conjugation.
\end{defin}

\begin{prop} \label{prop:symm1}
Let $R$ satisfies the tetrahedron equation and $\tau_{13}$ the permutation of the first
and third components, i.e $\tau_{13}(x,y,z)=(z,y,x)$. Then $\widetilde{R} = \tau_{13} \, R\,
\tau_{13}$ is also a solution of the tetrahedron equation.
\end{prop}

\begin{proof}
  The proof follows by direct check using the  identities
  \begin{eqnarray*}
R^{(123)} &=& \phantom{\tau_{12}}\,\tau_{16}\,\tau_{15}\,\tau_{25}\,\tau_{36} \, 
R^{(356)} \,\tau_{36}\,\tau_{25}\,\tau_{15}\,\tau_{16} \,, \\
R^{(145)} &=& \tau_{12}\,\tau_{26}\,\tau_{25}\,\tau_{15}\,\tau_{26} \, 
R^{(246)} \,\tau_{26}\,\tau_{15}\,\tau_{25}\,\tau_{26}\,\tau_{12} \,, \\
R^{(246)} &=& \phantom{\tau{12}\,\tau_{12}}\tau_{16}\,\tau_{25}\,\tau_{15} \, 
R^{(145)} \,\tau_{15}\,\tau_{25}\,\tau_{16} \,, \\  
R^{(356)} &=& \phantom{\tau_{12}}\,\tau_{26}\,\tau_{25}\,\tau_{15}\,\tau_{13} \, 
R^{(123)} \,\tau_{13}\,\tau_{15}\,\tau_{25}\,\tau_{26}  \,,
  \end{eqnarray*}
and the fact that $\tau_{ij}$ are involutions, i.e. $\tau_{ij}^2={\mathrm{Id}}$.
\end{proof}

\begin{prop} \label{prop:symm2}
Let $\sigma$ be an involutive symmetry of a tetrahedron map $R$, i.e.
\begin{equation}
\sigma \times \sigma \times \sigma \,  R  \, \sigma \times
\sigma \times \sigma = R \,, \qquad \sigma^2=  {\mathrm{Id}}  \,.
\end{equation}
Then the map
\begin{equation}
\widetilde{R} = \sigma \times {\mathrm{Id}} \times \sigma \, R  \,\, {\mathrm{Id}} \times
\sigma \times {\mathrm{Id}} \,,
\end{equation}
is a new solution of the tetrahedron equation. The same holds for the map 
$$\widehat{R} = {\mathrm{Id}} \times
\sigma \times {\mathrm{Id}} \,  R \, \sigma \times {\mathrm{Id}} \times \sigma \,. $$
\end{prop}
One may regard the previous statement as the tetrahedron analog of Proposition 2 in 
\cite{pstv2010} for the symmetries of Yang--Baxter maps. 

\begin{ex}
Consider the electric network transformation introduced in \cite{sergeev1998}, see also \cite{kks98}
\begin{equation} \label{star}
R(x,y,z) = \left(\frac{x\,y}{x+z+x\,y\,z}\, ,\, {x+z+x\,y\,z} \, , \, \frac{z\,y}{x+z+x\,y\,z}\right) \,.
\end{equation}
The involution $\sigma(x)=-x$, is a symmetry of the above map. Thus,
\begin{equation}
\widetilde{R}(x,y,z) = \left(\frac{x\,y}{x+z-x\,y\,z}\, ,\, {x+z-x\,y\,z} \, , \,
\frac{z\,y}{x+z-x\,y\,z}\right) \,,
\end{equation}
is another solution of the tetrahedron equation.
\end{ex}
In analogy with the Yang--Baxter maps we say that $R$ is reversible if it satisfies 
the following unitarity (reversibility) condition  
\begin{equation}
\tau_{13} \, R \, \tau_{13} \, R = \mathrm{Id}\,.
\end{equation}
For example, map \eqref{star} is both an involution and a reversible tetrahedron map. 

The existence of an involutive symmetry of a tetrahedron map can be exploited in 
different ways, such as obtaining solutions of entwining tetrahedron equations.
\begin{prop}
Let $\sigma$ be an involutive symmetry of a tetrahedron map $R$, then the following 
entwining tetrahedron equations hold
\begin{eqnarray*}
R^{(123)} \, {R^{(145)}\sigma_1} \, {R^{(246)}\sigma_2} \, {R^{(356)}\sigma_3} &=& 
{R^{(356)}\sigma_3} \, {R^{(246)}\sigma_2} \, {R^{(145)}\sigma_1} \, R^{(123)} \,,\\
{R^{(123)}\sigma_3} \, {R^{(145)}\sigma_5} \, {R^{(246)}\sigma_6} \, R^{(356)} &=&
R^{(356)}\,{R^{(246)}\sigma_6} \,{R^{(145)}\sigma_5} \, {R^{(123)}\sigma_3} \,,\\
{R^{(123)}\sigma_3} \, {R^{(145)}\sigma_1 \sigma_5} \, {R^{(246)}\sigma_2 \sigma_6} \, {R^{(356)}\sigma_3} &=& 
{R^{(356)}\sigma_3} \, {R^{(246)}\sigma_2 \sigma_6} \, {R^{(145)}\sigma_1 \sigma_5} \, {R^{(123)}\sigma_3} \,,  
\end{eqnarray*}
\end{prop}
\noindent where $\sigma_i$ acts as $\sigma$ on the $i$ component and as 
identity on the rest components. Similar entwining relations hold in the 
context of Yang--Baxter maps, as it was shown recently in \cite{Kassotakis:2019}.

\subsection{Invariants and symmetry groups of 3D lattice equations}

In this section we review some basic notions on the Lie symmetries and techniques for 
constructing absolute and relative invariants applied on 3D lattice equations. 
A more detailed exposition on the subject can be found in texts on invariant theory, 
see \cite{olver99}, \cite{olver95}.
We start with the building block of equations on the lattice 
$\mathbb{Z}^3 = \mathbb{Z}\, \mathbf{e}_1 + \mathbb{Z} \,\mathbf{e}_2 +\mathbb{Z}\, 
\mathbf{e}_3 $, 
consisting of a single algebraic relation of the form
\begin{equation}
\mathcal{E}(f,f_{ij},f_{ijk})=0\,, \quad i,j,k=1,2,3\,,\label{eq:quad} 
\end{equation}
where $f:\mathbb{Z}^3 \rightarrow \mathbb{C}$ (or $\mathbb{CP}^1$). 
We have reserved lower
indices to indicate the shifted values of $f$, 
$$ f_1=f((n_1+1)\,\mathbf{e}_1+n_2\,\mathbf{e}_2+n_3\,\mathbf{e}_3)\,,\quad 
f_{13}=f((n_1+1)\,\mathbf{e}_1+n_2\,\mathbf{e}_2+(n_3+1)\,\mathbf{e}_3)\,, $$
etc. where $n_1,n_2,n_3 \in \mathbb{Z}$, denote the site variables of the lattice. 
Examples of such lattice equations are given by the following list \cite{absocta}

\begin{eqnarray}
(\chi_1)  \quad    &  f_{12}\,f_3 - f_{13}\,f_2 +  f_{23}\,f_1  = 0 \,,   \\ [0.3cm]
(\chi_2)  \quad    &  \displaystyle{ \frac{ (f_{12}-f_{2}) \, (f_{23}-f_{3}) \,
                    (f_{13}-f_1)}{ (f_2  -f_{23}) \, (f_3-f_{13}) \, (f_1-f_{12})} } = -1 \,, \\ [0.3cm]
(\chi_3)  \quad    &   f_{12}\, (f_1 - f_2) + f_{23}\, (f_2 - f_3) + f_{13}\,(f_3 - f_1)=0\,,  \\ [0.3cm]
(\chi_4)  \quad    & \displaystyle{  \frac{f_{12} - f_{23}}{f_2}
                      +\frac{f_{23} - f_{13}}{f_3} + \frac{f_{13} - f_{12}}{f_1}=0\,, }\\ [0.3cm]
(\chi_5)  \quad    & \displaystyle{ \frac{f_{12}}{f_2} +  \frac{f_{23} - f_{13}}{f_3} -
                     \frac{f_{12}}{f_1} =0 \,. }
\end{eqnarray}
They are called of octahedron type since they relate the values of $f$ on 6 vertices of a three cube, 
which naturally define an octahedron. All lattice equations of the above list, even in non--commutative 
version, have appeared earlier in the literature with the corresponding PDE's names, 
see \cite{nc1990} and references therein. 
Their main common characteristic is the multi--dimensional consistency property, which will be explained 
in Section \ref{multi} in connection with discrete potential KP, equation denoted as $(\chi_3)$ above. 
It is exactly that particular property which is reflected to the tetrahedron property 
for the derived maps. We now turn to the basic concept of our investigations, namely the 
symmetry group of transformations of lattice equations and their corresponding differential invariants.

Consider a one--parameter group of transformations $G$ acting on the domain of the dependent variable $f$ 
\begin{equation}
G: f \mapsto \Phi(n_1,n_2,n_3,f;\varepsilon) \,, \quad \varepsilon \in \mathbb{C}\,. \nonumber
\end{equation}
Let ${\rm J}^{(k)}$ denote the lattice jet space with coordinates $(f,f_{J})$, where by
$f_{J}$ we mean the forward shifted values of $f$, indexed by all ordered multi--indices 
$J=(j^1,j^2,\ldots j^k)$, $1\leq j^r \leq 3$, of order $k=\# J$. 
The discrete prolongation of the group action of $G$ on ${\rm J}^{(k)}$ is
\begin{equation}
G^{(k)}: (f, f_J) \mapsto (\Phi(n_1,n_2,n_3,f;\varepsilon), 
\Phi_{J}(n_1,n_2,f;\varepsilon))\,,\label{eq:Gpr} 
\end{equation}
where $\Phi_{1}(n_1,n_2,n_3,f;\varepsilon)=\Phi(n_1+1,n_2,n_3,f_{1};\varepsilon)$, etc. namely
shifts of $\Phi$ on its arguments.
The infinitesimal generator of the group action of $G$ on $f$ is given by the vector field
\begin{equation}
\mathbf{v} = Q(n_1,n_2,n_3,f)\,\partial_{f}\,, \quad \mbox{where}\quad  
Q(n_1,n_2,n_3,f)=\left. 
\frac{\rm d \phantom{\varepsilon}}{{\rm d}\, \varepsilon}\,\Phi(n_1,n_2,n_3,f;\varepsilon) \right|_{\varepsilon=0} \,.
\nonumber
\end{equation}
There is a one--to--one correspondence between connected groups of transformations and their associated 
infinitesimal generators, since the group action is reconstructed by the flow of the vector 
field $\mathbf{v}$ by exponentiation
\begin{equation}
\Phi(n_1,n_2,f;\varepsilon)=\exp (\varepsilon \, \mathbf{v}) f \,.  \nonumber
\end{equation}
The infinitesimal generator of the action of $G^{(k)}$ on ${\rm J}^{(k)}$ is the 
associated $k^{\rm th}$ order prolonged vector field
\begin{equation}
\mathbf{v}^{(k)} = \sum_{\# J=j=0}^k Q_J(n_1,n_2,n_3,f)\,\partial_{f_J}\,. \nonumber
\end{equation}
The transformation $G$ is a Lie--point symmetry of the lattice equations (\ref{eq:quad}), 
if it transforms any solution of (\ref{eq:quad}) to another solution of the same equation. 
Equivalently, $G$ is a symmetry of equation (\ref{eq:quad}), if the equations are not 
affected by the transformation (\ref{eq:Gpr}). 
The infinitesimal criterion for a connected group of transformations $G$ to be a symmetry 
of equation (\ref{eq:quad}) is 
\begin{equation}
\mathbf{{v}}^{(3)} \big(\mathcal{E}(f,f_{i},f_{ij},f_{ijk}\big) =0\,. \label{eq:infinv}
\end{equation}
Equation (\ref{eq:infinv}) should hold for all solutions of equations (\ref{eq:quad}), 
and thus the latter equation and its shifted consequences should be taken into account. 
Equation (\ref{eq:infinv}) determines the most general infinitesimal Lie point
symmetries of the (\ref{eq:quad}). The resulting set of infinitesimal
generators forms a Lie algebra $\mathfrak{g}$ from which the corresponding 
Lie point symmetry group $G$ can be found by exponentiating the given vector fields.

An (absolute) {\em lattice invariant} under the action of $G$ is a function 
$I:{\rm J}^{(k)}\rightarrow \mathbb{C}$ 
which satisfies 
\begin{equation} 
  I(g^{(k)} \cdot (f, f_J)) = I( f,f_J)\,,
\end{equation}  
for all $g \in G$, and all $( f,f_j) \in {\rm J}^{(k)}$.
For connected groups of transformations, a necessary and sufficient condition for a function 
$I$ to be invariant under the action of $G$, is the annihilation 
of $I$ by all prolonged infinitesimal generators, i.e.
\begin{equation} \label{ic}
\mathbf{{v}}^{(k)}(I) = 0 \,, 
\end{equation}
for all $\mathbf{{v}} \in \mathfrak{g}$. The result is an overdetermined system of linear,
homogeneous partial differential equations which, in principle, can be solved by the method 
of characteristics. 
The conditions under which such invariants exist are specified by Frobenius Theorem. 

Using the dual formulation of Frobenius Theorem in terms of differential forms, 
the problem of finding the invariants under a symmetry group $G$ can be recast to that of
solving a homogeneous linear system. This is accomplished by the natural correspondence  
between the vector $\mathbf{v}$ and the differential ${\mathrm{d}} I$ of the function $I$ given by
\begin{equation} \label{ic1}
  \langle {\mathrm{d}} I \, , \mathbf{v} \rangle \equiv \mathbf{v}^{(k)}(I) = 
  \sum_{\# J=j=0}^k Q_j\,\partial_{f_j} \, (I) \,,
\end{equation}
where $\langle \,\,,\,\rangle$ denotes the usual inner product between a vector and a dual vector.
Combining \eqref{ic} and \eqref{ic1} we conclude that the differential ${\mathrm{d}}I$ of a
solution of \eqref{ic} annihilates the vector $\mathbf{v}$. 
Moreover, the latter is true for any 1--form $\omega$ which is a non zero scalar multiple of the 
differential ${\mathrm{d}}I$, i.e.
$$ \omega = \alpha \, {\mathrm{d}}I \,. $$
This observation suggests an alternative way for constructing a complete set of functionally
independent invariants of a symmetry group $G$.
Consider $m$ independent vector fields $\mathbf{v}_1,\mathbf{v}_2 \, \ldots, \mathbf{v}_m$, where 
$$ \mathbf{v}_\mu = \sum_{\# J=j=0}^k Q_{j,\mu}\,\partial_{f_j}  \,, \qquad \mu=1,2,\ldots m\,,$$
which generate a connected $m$--dimensional Lie group $G$.
We are looking for functions $I:{\rm J}^{(k)}\rightarrow \mathbb{C}$  which satisfy the 
invariance conditions 
$$\mathbf{v}^{(k)}_1(I)=0\,, \ldots, \mathbf{v}^{(k)}_m(I)=0 \,, $$
for fixed order $k$. In the language of forms the problem is interpreted to determine a 
set of $k-m$ independent 1--forms 
$$ \omega_{\nu} = \sum_{\# J=j=0}^k \omega_{ j,\nu} \, {\rm d} f_j \,, \qquad 
\nu = 1,2,\ldots , k-m\,, $$
which annihilate these vectors fields, i.e. 
$\langle \omega_\nu \, , \mathbf{v}_\mu \rangle =0 $. The latter condition is equivalent 
to the following linear system of equations
\begin{equation} \label{linsys}
 \sum_{j=0}^k Q_{j,\mu} \,  \omega_{ j,\nu} =0 \,,  \, 
\end{equation}
for $\mu=1,2,\ldots m$, and $\nu = 1,2,\ldots , k-m$. The next step is to determine integrating
factors $\alpha_{\nu\lambda}$, and functions $I_\nu$, with $\nu,\lambda=1,2,\ldots , k-m$,
such that 
$$ \omega_\nu = \sum_{\lambda=1}^{k-m} \alpha_{\nu\lambda} \, {\mathrm d} I_\lambda \,,$$
and $\det (\alpha_{\nu\lambda}) \neq 0$. Then \eqref{ic1} implies that the $k-m$ 
functions $I_\lambda$ form a complete set of functionally independent (absolute) invariants of 
the Lie symmetry group $G$ generated by the vector fields $\{\mathbf{v}_\mu\}$. 

Relaxing the condition of absolute invariance under a symmetry group, for some vector fields,
we arrive at the
notion of relative, or semi--invariants, i.e. functions  $I:{\rm J}^{(k)}\rightarrow \mathbb{C}$  
which satisfy 
\begin{equation} \label{icsemi}
\mathbf{{v}}_i^{(k)}(I) = \beta_i \, I \,, 
\end{equation}
where each $\beta_i$ is called the index of the semi--invariant. 
Invariants are semi invariants of index 0. 

To be more specific we consider now both kinds of invariants under the symmetry group of
lattice AKP and BKP equations. Lattice AKP and BKP are the following 
discrete equations
$$ -f_1 \, f_{23} + f_2 \, f_{13} - f_3 \, f_{12} = 0  \,, \qquad 
f \, f_{123} -f_1 \, f_{23} + f_2 \, f_{13} - f_3 \, f_{12} = 0 \,,$$
respectively. Both equations admit the same (abelian) Lie symmetry group $G$ generated by 
the vector fields 
\begin{equation} \label{symakp}
 \mathbf{v}_1 = f\,\partial_f \,, \qquad \mathbf{v}_2 = (-1)^{n_1+n_2+n_3} 
\, f \,\partial_f \,.
\end{equation}
The corresponding symmetry transformations are the usual and alternating scalings 
$$ f \mapsto e^{\varepsilon} \, f\,,\qquad f \mapsto e^{[(-1)^{n_1+n_2+n_3}]\varepsilon} \, f\,.$$
respectively. 
Moreover, both discrete equations remain invariant under the action of the discrete symmetry 
$$D:f\mapsto (-1)^{n_1+n_2+n_3} \, f \,.$$
In the following we denote the alternating factor $(-1)^{n_1+n_2+n_3}$ by $\alpha \equiv (-1)^{n_1+n_2+n_3}$.
 
Let us restrict on the face of a three cube where the four fields $f,f_1,f_2,f_{12}$ are 
assigned (see Fig. \ref{fig:cube} below).
Inserting the second order prolongation of the vector fields
$$ \mathbf{v}_1^{(2)} = f\,\partial_f + f_1 \partial_{f_1} +  f_2 \partial_{f_2} + f_{12} \partial_{f_{12}} \,, \qquad 
\mathbf{v}_2^{(2)} = \alpha \, (f\,\partial_f - f_1 \partial_{f_1} - f_2 \partial_{f_2} + f_{12} \partial_{f_{12}}) \,,$$ 
to the linear system \eqref{linsys} we obtain  
\begin{equation}
\left(
  \begin{array}{rrrr} 
  f & f_1 & f_2 & f_{12} \\
  f & -f_1 & -f_2 & f_{12} \end{array} \right) 
\left(
  \begin{array}{rr} 
  {\boldsymbol{\Omega}_{1}}^t & {\boldsymbol{\Omega}_{2}}^t \\
  \end{array} \right) = 
\left(
  \begin{array}{cc} 
  0 & 0 \\
  0 & 0  \end{array} \right)   \,, 
\end{equation}
where $\boldsymbol{\Omega}_{i} = (\omega_{0,i},\omega_{1,i},\omega_{2,i},\omega_{12,i})$, $i=1,2$. 
Using a simple row reduction on the preceding linear system we get 
\begin{equation}
  \boldsymbol{\Omega}_{1} = (0,f_2,-f_1,0) \,, \qquad 
  \boldsymbol{\Omega}_{2} = (-f_{12},0,0,f) \,,
\end{equation}
and the two linearly independent 1-forms $\omega_1,\omega_2$ take the form
\begin{equation}
\omega_1 = f_2\,{\mathrm d} f_1 - f_1\,{\mathrm d} f_2  = 
{(f_2)^2} \, {\mathrm d} \left( \frac{f_1}{f_2} \right) \,, \quad 
\omega_2 = -f_{12}\,{\mathrm d} f + f \,{\mathrm d} f_{12}  = 
{f^2} \, {\mathrm d} \left( \frac{f_{12}}{f} \right) \,,
\end{equation}
respectively. Hence, 
\begin{equation} \label{invakp}
I_1=\frac{f_1}{f_2}\,, \qquad I_2=\frac{f_{12}}{f}\,,
\end{equation}
form a complete set of (absolute) invariants under the action of $G$ on the
plaquette of the cube where $(f,f_1,f_2,f_{12})$ are assigned. Note that $I_1$ and $I_2$
remain invariant under the action of the discrete symmetry $D$, as well, $D(I_i)=I_i$, $i=1,2$.

Next, we search for semi--invariants satisfying the relations 
$\mathbf{v}_1(S)=0$, $\mathbf{v}_2(S)=\beta\,S$, for some function $\beta$ independent of $f$ and 
its shifts. By similar considerations as above, a complete set of semi--invariants under 
$G$ on ${\mathrm{J}}^{(2)}=(f,f_1,f_2,f_{12})$ satisfying $
\mathbf{v}_1(S)=0$,  $\mathbf{v}_2(S)=2\,\alpha \, S$, is given by 
$$ S_1 = \frac{f}{f_1} \,, \qquad S_2 = \frac{f_{12}}{f_2}\,. $$
If we further demand (absolute) invariance under the action of the discrete symmetry 
$D$, then we can combine $S_1$, and $S_2$ to obtain the following semi--invariant 
\begin{equation} \label{semiakp}
 S_3= \frac{f\,f_{12}}{f_1 \, f_2} \,, 
 \end{equation}
on ${\mathrm{J}}^{(2)}$, which satisfies $\exp (\varepsilon \, \mathbf{v}_1) S_3=S_3=D(S_3)$, 
and $\exp (\varepsilon \, \mathbf{v}_2)  S_3=e^{4\,\alpha\,\varepsilon} \, S_3$.

\section{Tetrahedron maps from integrable 3D lattice equations} \label{class}
\subsection{AKP ($\chi_1$) and BKP}
We consider first the case of (absolute) $G$--invariants of AKP and BKP lattice equation. 
As mentioned above AKP, or equation denoted as $\chi_1$ in classification \cite{absocta}, reads 
\begin{equation}
 -f_1 \, f_{23} + f_2 \, f_{13} - f_3 \, f_{12} = 0  \,,
\end{equation}
and admits the Lie symmetry group $G$ generated by the vector fields \eqref{symakp}.
On each elementary plaquette of the 3 cube, for example the one with vertices 
$V_{12}=\{f,f_1,f_2,f_{12}\}$, the number of functionally independent (absolute) 
invariants is two, since the action of $G$ on the space of four variables of $V_{12}$ has 
two dimensional orbits. Choosing an orientation, we take as $G$--invariants the diagonals 
\eqref{invakp}, on all faces of the cube, namely
\begin{equation} \label{inv_ratios}
\begin{array}{*3{>{\displaystyle}l}}
(x_1,y_1) = \left( \frac{f_1}{f_2},\frac{f_{12}}{f} \right) \,, & 
(x_2,y_2) = \left( \frac{f_{12}}{f_{23}}, \frac{f_{123}}{f_2} \right)\,, &
(x_3,y_3) = \left( \frac{f_2}{f_3}, \frac{f_{23}}{f} \right) \,, \\ [0.5cm]
(u_1,v_1) = \left( \frac{f_{13}}{f_{23}},\frac{f_{123}}{f_3} \right)\,, &
(u_2,v_2) = \left( \frac{f_{1}}{f_{3}}, \frac{f_{13}}{f} \right)\,, &
(u_3,v_3) = \left( \frac{f_{12}}{f_{13}}, \frac{f_{123}}{f_1} \right) \,.
\end{array}
\end{equation}
The above invariants are not functionally 
independent, since the Jacobian matrix, has rank six. 
On the other hand, since $G$ has two dimensional orbits, the number of 
functionally independent invariants on ${\mathrm J}^{(3)}$ is precisely $8-2=6$, the number of
vertices of the three cube minus the dimension of $G$--orbits. Thus, our aim is to construct a set
of six functionally independent relations among the twelve invariants \eqref{inv_ratios}, 
which will serve to characterise a complete set of functionally independent invariants on 
${\mathrm J}^{(3)}$. Eliminating the $f_J$'s among the defining relations \eqref{inv_ratios}, we
obtain the following six functionally independent relations:
\begin{equation} \label{akp_set1}
y_1 = x_2\,y_3  \,, \quad v_1 = u_2\,v_3 \,, \quad v_1 = x_3\,y_2 \,, 
\quad v_2 = y_3\, u_1 \,,\quad x_2=u_1\,u_3\,, \quad u_2=x_1\,x_3\,.
\end{equation}
Moreover, since $G$ is a symmetry group of AKP, the latter lattice equation can be written in 
terms of the invariants $I^i$, namely
\begin{equation} \label{akp_set2}
  x_2 + u_2 = u_1\, x_3 \,.
\end{equation} 
Now we view the first relation of the set \eqref{akp_set1} as a compatible constraint and consider 
the system formed by the rest five relations from \eqref{akp_set1} together with \eqref{akp_set2}. 
Solving the system for $(u_i,v_i)$, we obtain the unique solution
\begin{equation} \label{ft_akp1}
\begin{array}{*3{>{\displaystyle}l}}
u_1 = \frac{(x_2 + x_1\, x_3)}{x_3} \,, & u_2 = x_1 \, x_3\,, & 
u_3 = \frac{x_2\, x_3}{x_2 + x_1\, x_3} \,, \\ [0.5cm]
v_1 = x_3\, y_2\,, & v_2 = \frac{(x_2 + x_1\, x_3)\,y_3}{x_3}\,, & 
v_3 = \frac{y_2}{x_1}\,. \end{array}
\end{equation}      
The solution can be interpreted as a map 
$R:\left( {x_1\choose y_1} , {x_2\choose y_2} , {x_3\choose y_3} \right) \mapsto 
\left( {u_1\choose v_1} , {u_2\choose v_2} , {u_3\choose v_3} \right)$. 
The 4D consistency property  of AKP, and the compatible symmetry invariants, 
imply that the map $R$ satisfies the functional tetrahedron relation.

\begin{figure}[h!]
  \centering
  \def\svgscale{0.4}
  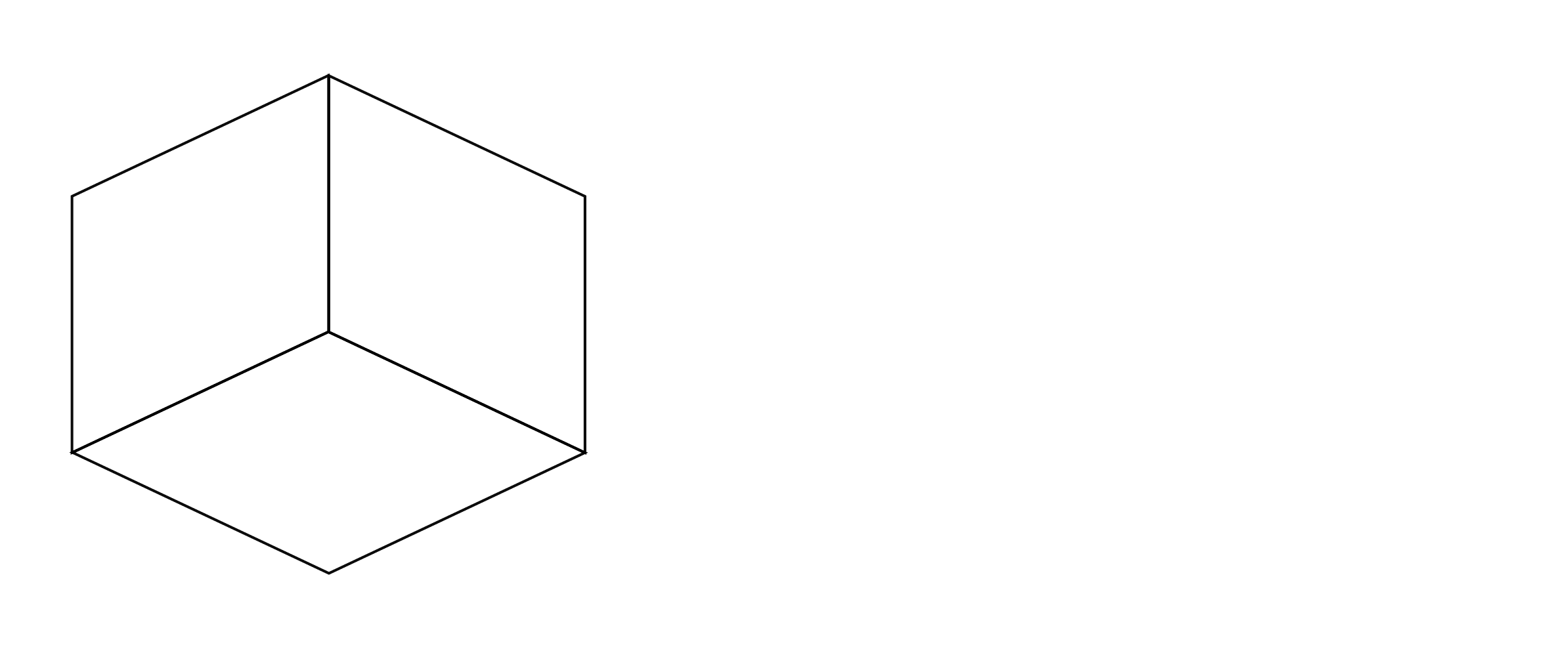
    \caption{$R:\left( {x_1\choose y_1} , {x_2\choose y_2} , {x_3\choose y_3} \right) \mapsto 
\left( {u_1\choose v_1} , {u_2\choose v_2} , {u_3\choose v_3} \right)$}
    \label{fig:cube}
\end{figure}

We note that the above map $R$ has a triangular form in which the $x$-component
$R^x : (x_1,x_2,x_3) \mapsto (u_1,u_2,u_3)$ decouples from the $y$ components of the map. 
The map $R^x$ is map (19) in the classification obtained by Sergeev \cite{sergeev1998},
thus tetrahedron map \eqref{ft_akp1} can be regarded as a vector extension of it.
Moreovet, $R^x$ is reversible i.e. it satisfies the relation 
$\tau_{13}\, R^x \, \tau_{13} \, R^x = \mathrm{Id}$. 

The map $R$, as it stands in \eqref{ft_akp1}, satisfies the tetrahedron equation without taking 
into account any compatible constraint. However, $R$ is not invertible. 

One can obtain a map $\widetilde{R}$ by invoking the invariant relations \eqref{akp_set1}, 
but this time taking into account the leftover relation $y_1=x_2\,y_3$, and at the same time 
dropping out the second relation $v_1=u_2\,v_3$ from the set \eqref{akp_set1}. 
Then we can solve the latter set together with \eqref{akp_set2}, uniquely, 
to get $\widetilde{R}$ defined by
\begin{equation} \label{Rakpinv}
\begin{array}{*3{>{\displaystyle}l}}
x_1 = \frac{u_1 \, u_2}{u_2 + u_1 \, u_3} \,, & x_2 = u_1 \, u_3\,, & 
x_3 = \frac{u_2 + u_1\,u_3}{u_1} \,, \\ [0.5cm]
y_1 = u_3\, v_2\,, & y_2 = \frac{u_1\,v_1}{u_2 + u_1\, u_3}\,, & 
y_3 = \frac{v_2}{u_1}\,. \end{array}
\end{equation} 
Then 
$$\widetilde{R} \circ R\, (x_1,y_1,x_2,y_2,x_3,y_3)=(x_1, x_2\, y_3, x_2, y_2, x_3, y_3)$$ 
which is the identity  map by taking into account the relation $y_1=x_2\,y_3$. Similarly, 
$$R \circ \widetilde{R}(u_1,v_1,u_2,v_2,u_3,v_3)=(u_1, v_1, u_2, v_2, u_3, \frac{v_1}{u_2})$$ 
which is again the identity map by taking into account the second relation in \eqref{akp_set1}, 
i.e. $v_1 = u_2\,v_3$. The reasoning behind this fact is that AKP satisfies the 4D consistency 
property, and its symmetry invariants are by definition compatible with AKP, hence any constraint 
among the $G$--invariants propagate on the whole $\mathbb{Z}^3$ in a consistent way, 
allowing one to solve the functional relations among the invariants in both directions 
$R:(x_i,y_i) \leftrightarrows (u_i,v_i):\widetilde{R}$. Moreover, by embedding the lattice equation 
and their compatible $G$--invariants on $\mathbb{Z}^4$, then the consistency property of the lattice 
equation implies that both $R$ and $\widetilde{R}$ satisfy the tetrahedron relation. 
To be more concrete, notice that $R$ maps the symmetry constraint $y_1=x_2\,y_3$, on the 
left half cube of Fig. \ref{fig:cube}, to $v_1=u_2\,v_3$, which lives on the right half cube, 
and can be verified readily from \eqref{ft_akp1}, and vice versa from \eqref{Rakpinv}. 

Next we consider BKP lattice equation which reads 
\begin{equation}
f\,f_{123} - f_1 \, f_{23} + f_2 \, f_{13} - f_3 \, f_{12} = 0  \,.
\end{equation}
Since BKP and AKP share the same symmetries we take the same 
$G$-invariants \eqref{inv_ratios}, together with the corresponding functional relations
\eqref{akp_set1}. In terms of its $G$-invariants BKP now reads
\begin{equation} \label{bkpinv}
1 - \frac{y_3}{v_3} + \frac{v_2}{y_2} - \frac{y_1}{v_1}=0  \,,
\end{equation}
and from the unique solution of the corresponding system we obtain the tetrahedron map
$R:\left( {x_1\choose y_1} , {x_2\choose y_2} , {x_3\choose y_3} \right) \mapsto 
\left( {u_1\choose v_1} , {u_2\choose v_2} , {u_3\choose v_3} \right)$ defined by the relations
\begin{equation} \label{ft_bkp1} 
\begin{array}{*3{>{\displaystyle}l}}
u_1 = x_1 + \frac{y_1}{x_3\,y_3} - \frac{y_2}{y_3}, & u_2 = x_1 \, x_3\,, & 
u_3 = \frac{x_2\, x_3\,y_3}{y_1-y_2\, x_3 + x_1\, x_3 \, y_3} \,, \\ [0.5cm]
v_1 = x_3\, y_2\,, & v_2 = \frac{y_1}{x_3} + x_1\, y_3 - y_2\,, & v_3 = \frac{y_2}{x_1}\,. \end{array}
\end{equation} 
The map is coupled and the tetrahedron property is satisfied if
we take into account the symmetry constraint $y_1=x_2\,y_3$ on the auxiliary space $(123)$ and the
corresponding constraint on (356), reflecting the fact that $R$ is a tetrahedron map modulo the 
symmetry constraint. 

We now consider the semi--invariants of AKP and BKP. In both cases we take on all faces of 
the cube the following semi--invariants
\begin{equation} \label{semi}
s_1=\frac{f\,f_{12}}{f_1\,f_2}\,,\quad
s_2=\frac{f_2\,f_{123}}{f_{12}\,f_{23}}\,,\quad s_3=\frac{f\,f_{23}}{f_2\,f_3}
\,, \quad t_1=\frac{f_3\,f_{123}}{f_{13}\,f_{23}} \,, \quad  
t_2=\frac{f\,f_{13}}{f_{1}\,f_{3}}\,,\quad t_3=\frac{f_1\,f_{123}}{f_{12}\,f_{13}}\,,  
\end{equation} 
which are functionally related by 
\begin{equation} \label{funct}
 s_1\, s_2 = t_1\,t_2\,,\qquad s_2\,s_3=t_2\,t_3 \,.
\end{equation}
The crucial point is that AKP can be written in terms of the above semi--invariants as follows
\begin{equation}  \label{okp} 
t_2 = s_1+s_3 \,.
\end{equation}
Solving (\ref{funct}), (\ref{okp}) for
$(t_1,t_2,t_3)$ we obtain the tetrahedron map $R:(s_1,s_2,s_3)\mapsto (t_1,t_2,t_3)$, defined by
\begin{equation} \label{octakp}
t_1=\frac{s_1\,s_2}{s_1+s_3}\,,\qquad t_2= s_1+s_3 \,,\qquad t_3=\frac{s_2\,s_3}{s_1+s_3}\,.
\end{equation}
The above map appears also in the classification \cite{sergeev1998} by Sergeev, as a limit of some 
other tetrahedron map, and is both an involution and a reversible tetrahedron map. 

On the other hand, BKP can be written in terms of the semi--invariants \eqref{semi} as
\begin{equation} \label{bkpinv2}
  s_1 \, s_2 \, s_3 - s_2 + t_2 -s_1 =0 \,,
\end{equation}
and from \eqref{funct}, \eqref{bkpinv2} we obtain the well--known electric network transformation
\begin{equation}
t_1=\frac{s_1\,s_2}{s_1+s_3 - s_1\,s_2\,s_3}\,,\qquad t_2= s_1+s_3 - s_1\,s_2\,s_3\,,
\qquad t_3=\frac{s_2\,s_3}{s_1+s_3-s_1\,s_2\,s_3}\,,
\end{equation}
The above map appeared first in \cite{kas1996} and as was mentioned earlier is both an involution 
and a reversible tetrahedron map.

\subsection{Schwarzian KP ($\chi_2$)}
Equation $\chi_2$ in ABS classification, or the generalised lattice spin equation (Nijhoff and Capel) 
is the following lattice equation
\begin{equation}  \label{latticex2}
\frac{f_1-f_{12}}{f_{12}-f_2} \,\, \frac{f_2-f_{23}}{f_{23}-f_3} \,\, \frac{f_3-f_{13}}{f_{13}-f_1} = -1 \,.
\end{equation}
It is invariant under fractional linear (M\"obius) transformations 
$$f\mapsto \frac{a\,f+b}{c\,f+d}\,,$$ 
with the corresponding infinitesimal symmetry generators of the $sl(2)$ symmetry being
$$X_1=\partial_f\,,\qquad X_2 = f\,\partial_f \,, \qquad X_3=f^2\,\partial_f\,.$$ 
The transformation $f \mapsto f^{-1}$ maps $X_1$ to $X_3$, thus essentially  there exists only one 
two--dimensional symmetry subgroup $H$, the one generated by the subalgebra $\{X_1,X_2 \}$. 
We choose on every face of the 3 cube the following invariants $I_i$ along the $H$--orbits 
\begin{equation}
\begin{array}{*3{>{\displaystyle}l}} \label{invS1}
(x_1,y_1) = \big( \frac{f-f_2}{f-f_1},\frac{f_1-f_{12}}{f_2-f_{12}} \big) \,, & 
(u_1,v_1) = \big( \frac{f_{3}-f_{23}}{f_{3}-f_{13}},\frac{f_{13}-f_{123}}{f_{23}-f_{123}} \big)\,, \\ [0.5cm]
(x_2,y_2) = \big( \frac{f_2-f_{23}}{f_2-f_{12}}, \frac{f_{12}- f_{123}}{f_{23}-f_{123}} \big)\,, &
(u_2,v_2) = \big( \frac{f-f_3}{f-f_1} , \frac{f_1-f_{13}}{f_3-f_{13}} \big)\,,  \\ [0.5cm]
(x_3,y_3) = \big( \frac{f-f_3}{f-f_2}, \frac{f_2-f_{23}}{f_3-f_{23}} \big) \,, &
(u_3,v_3) = \big( \frac{f_1-f_{13}}{f_1-f_{12}} , \frac{f_{12}-f_{123}}{f_{13}-f_{123}} \big) \,.
\end{array}
\end{equation}
The rank of the Jacobian matrix $(\partial I_i/\partial f_j)$ is $6$, hence eliminating the 
$f_j$'s among the defining relations of the above invariants, we obtain the following six 
functionally independent relations:
\begin{equation} \label{x2_set1}
\begin{array}{lll} 
  x_2 (1 - \frac{1}{x_1} )  (1-\frac{1}{y_3}) = (1-x_3) (1-y_1)\,, 
& y_2 = v_1 \, v_3 \,, 
& y_3  (1 - \frac{1}{x_2} )  (1-v_1) = (1-y_2) (1-\frac{1}{u_1})\,, \\ [0.5cm]
  v_2  (1 - \frac{1}{v_1} )  (1-\frac{1}{u_3}) = (1-v_3) \, (1-u_1)\,,
& u_2 = x_1 \, x_3 \,, 
& u_3 (1 - \frac{1}{v_2} ) (1-x_1) = (1-u_2) (1-\frac{1}{y_1})\,. \\ [0.5cm] 
\end{array}
\end{equation}
Lattice equation \eqref{latticex2} can be written in terms of its symmetry invariants in many ways, such as
\begin{equation} \label{x2_set2}
\begin{array}{*2{>{\displaystyle}l}}
v_2 = y_1 \, y_3 \,, & (1 - \frac{1}{x_2} ) (1-\frac{1}{x_3}) = (1-\frac{1}{u_2}) (1-\frac{1}{u_3})\,, \\ [0.5cm]
x_2 = u_1 \, u_3 \,, & (1 - \frac{1}{y_2} ) (1-\frac{1}{y_3}) = (1-\frac{1}{v_2}) (1-\frac{1}{v_3})\,. \\
\end{array}
\end{equation}
Again, we view the first relation from the set of equations \eqref{x2_set1} as a 
compatible constraint and consider the remaining equations together with any of the 
invariant forms of the lattice equation from the set
\eqref{x2_set2}. Taking into account the compatible constraint \eqref{x2_set1}(i)
all corresponding systems can be solved uniquely for $(u_i,v_i)$, 
leading to the same solution
\begin{equation} \label{ftx2}
\begin{array}{*3{>{\displaystyle}l}}
u_1 = \frac{x_1 + x_2  - x_1\,x_2 - x_1\,x_3}{1-x_1\,x_3}, & 
u_2 = x_1 \, x_3\,, & 
u_3 = \frac{x_2\,(1-x_1\,x_3)}{x_2 + x_1  - x_1\,x_2 - x_1\,x_3} \,, \\ [0.5cm]
v_1 = \frac{y_1 + y_2  - y_1\,y_2 - y_1\,y_3}{1-y_1\,y_3},   &
v_2 = y_1 \, y_3\,, &
v_3 = \frac{y_2\,(1-y_1\,y_3)}{y_2 + y_1  - y_1\,y_2 - y_1\,y_3} \,. \end{array} 
\end{equation}
The above map decouples into two identical involutions $R = R^x \times R^y$. 
Conjugating $R^x$ by $x\mapsto 1-x$
the resulting map becomes the tetrahedron map (25) in Sergeev's classification 
\cite{sergeev1998}. 

There exists also a degeneration of the lattice equation \eqref{latticex2} which is 
multi dimensional consistent and was introduced by Nijhoff and Capel in \cite{nc1990}, 
under the name generalized lattice modified Toda equation. It reads the form
\begin{equation} \label{latticex2b}
f_2\,(f_2-f_{23})^{-1}\,(f_{23}-f_3) = (f_{12}-f_1)\,(f_1-f_{13})^{-1}\,f_3 \,,
\end{equation} 
and is obtained form (\ref{latticex2}) by considering first a gauge transformation 
$f\mapsto \alpha^n\,\beta^m\,\gamma^k \, f$ and then taking $\alpha \rightarrow 0$ 
and $\beta=\gamma=1$.
In this case, only the scaling symmetry is admitted by the lattice equation.
However, it is interesting to note that the degenerate invariants under the
remaining scaling symmetry can still be used, in a consistent way, 
to obtain a tetrahedron map which it turns out to be again
the tetrahedron map \eqref{ftx2} with vanishing $y$ and $v$ variables.

\subsection{Discrete potential KP ($\chi_3$)} \label{kpx3}
Equation $\chi_3$ in ABS classification, or in the non--commutative case, the
discrete potential KP (dpKP) 
\begin{equation} \label{dpKP}
 f_{12}\, (f_1 - f_2) + f_{23}\, (f_2 - f_3) + f_{13}\, (f_3 - f_1)=0 \,.
\end{equation}
The dpKP can be written in the following symmetric form
\begin{equation} \label{KP_symm}
  \frac{f_{12}-f_{23}}{f_1-f_3} = \frac{f_{13}-f_{23}}{f_1-f_2} = \frac{f_{12}-f_{13}}{f_2-f_3} \,.
\end{equation}
It is invariant along the orbits of the Lie group $G$ generated by the symmetry vector
fields
$$ X_1=\partial_f\,, \quad X_2=(-1)^{n+m+k} \partial_f\,,  \quad 
X_3= f\,\partial_f\,, \quad X_4= (-1)^{n+m+k} \,f\,\partial_f \,. $$  
The non--vanishing commutators of the infinitesimal generators are
$$ [X_1,X_3]= X_1\,, \quad [X_1,X_4]= X_2 \,, \quad [X_2,X_3]=X_2\,,\quad [X_2,X_4]=X_1 \,, $$
spanning a Lie algebra $\frak{g}$, classified already by Sophus Lie \cite[pp. 727]{lie1893}.
The corresponding Lie group $G$ is isomorphic to the solvable group of complex affine transformations 
of the complex line
$$ G \cong {\rm{Aff}}(\mathbb{C}) = \{ z \mapsto a z+b \,\,  |\,\, \,(a,b)\in
\mathbb{C}^{\ast}\times \mathbb{C} \} \,.$$ 
Lie group $G$ has four two--dimensional subgoups generated by the subalgebras 
$A_1=\{X_1,X_2\}$, $A_2=\{X_3,X_4\}$,
(abelian) and $S_1=\{X_1,X_3\}$, $S_2=\{X_2,X_3\}$ (solvable). However, the
discrete symmetry $D: f \mapsto (-1)^{n+m+k} \, f$ of lattice potential KP \eqref{dpKP}, maps the
infinitesimal generators spanning subalgebra $S_1$ to those spanning $S_2$. 
Thus, up to the discrete symmetry $D$, there exist three nonequivalent two dimensional subalgebras 
namely, $A_1, A_2$, and $S_1$. In the following, for each subalgebra, we derive the corresponding 
tetrahedron maps.

\vspace{0.5cm}\noindent
{\em Case a:} 
We choose invariants $I_i$ along the orbits of the subgroup generated by $A_1$, the diagonals
\begin{equation} \label{translation} 
\begin{array}{*3{>{\displaystyle}l}}
(x_1,y_1) = \big( {f_1}-{f_2}\,,\,{f_{12}}-{f} \big)\,, & 
(u_1,v_1) = \big( {f_{13}}-{f_{23}}\,,\,{f_{123}}-{f_3} \big) \,,  \\ [0.5cm]
(x_2,y_2) = \big( {f_{12}}-{f_{23}}\,,\, {f_{123}}-{f_2} \big)\,, &
(u_2,v_2) = \big( {f_{1}}-{f_{3}}\,,\,{f_{13}}-{f} \big) \,, \\ [0.5cm]
(x_3,y_3) = \big( {f_2}-{f_3}\,,\, {f_{23}}-{f} \big)\,, &
(u_3,v_3) = \big( {f_{12}}-{f_{13}}\,,\, {f_{123}}-{f_1} \big) \,.
\end{array}
\end{equation}
The rank of the Jacobian matrix $(\partial I_i/\partial f_j)$ is $6$, thus we obtain the 
following six functionally independent relations
\begin{equation}\label{x3_set1}
y_1=x_2+y_3\,,\quad v_1=u_2+v_3\,,\quad y_1 = u_3+v_2\,,\quad y_2=x_1+v_3\,,\quad  
x_2=u_1+u_3\,,\quad u_2=x_1+x_3\,.
\end{equation}
In terms of the above invariants dpKP is written as
\begin{equation} \label{x3_inv}
 \frac{x_2}{u_2}\, =\, \frac{u_1}{x_1}\, =\, \frac{u_3}{x_3}\,. 
\end{equation}
Solving the set of equation formed by \eqref{x3_set1}, except the first one, and either 
equation from \eqref{x3_inv} we obtain the unique solution
\begin{equation} \label{x3a}
\begin{array}{*3{>{\displaystyle}l}}
 u_1 = \frac{x_1\,x_2}{x_1+x_3}, & u_2 = x_1 + x_3\,, & u_3 = \frac{x_2\, x_3}{x_1+ x_3}\,, \\ [0.5cm]
 v_1 = x_3+ y_2\,,  & v_2 = \frac{x_1\,y_1+x_3\,y_3}{x_1+x_3} \,,& v_3 = {y_2}-{x_1}\,. \end{array} 
\end{equation}
The map is an involution (taking into account the constraint $y_1=x_2+y_3$) and has a triangular form, 
in which the $R^x$ component coincides with the tetrahedron map \eqref{octakp}.

\vspace{0.5cm}\noindent
{\em Case b:} Next, we consider the abelian subgroup generated by $A_2=\{X_3,X_4\}$ and its
associated invariants given by \eqref{inv_ratios}. 
The symmetric form \eqref{KP_symm} of dpKP equation suggests that it can can be written 
in terms of the above invariants, as
\begin{equation} \label{kp_symm_b}
  \frac{y_1-y_3}{u_2-1} \, = \, \frac{v_2-y_3}{u_2-x_3} \, = \, \frac{y_1-v_2}{x_3-1} \,.
\end{equation}
As in the previous case, dropping the first of equations \eqref{inv_ratios} and 
considering the system of the remaining equations together with any of equations 
\eqref{kp_symm_b} we obtain the unique solution
\begin{equation}
\begin{array}{*3{>{\displaystyle}l}}
 u_1 = \frac{Q}{1-x_1\,x_3} \,,  &
 u_2 = x_1\,x_3 \,, & 
 u_3 = \frac{(1-x_1\,x_3)\,x_2}{Q}\,, \\ [0.5cm]
 v_1 = x_3\,y_2\,,  & 
 v_2 = \frac{Q \,y_1}{(1-x_1\,x_3)\,x_2} \,, & 
 v_3 = \frac{y_2}{x_1} \,, \end{array} 
\end{equation}
where $ Q= 1 - x_3  + x_2\,x_3 - x_1\,x_2\,x_3$. 
The map again has a triangular form. The conjugation $x\mapsto 1+x$ transforms the 
$x$-component of the above tetrahedron map to 
\begin{equation} \label{x3b}
R^x : (x_1,x_2,x_3) \mapsto 
 \left( \frac{x_1\,x_2\,(1+x_3)}{x_1 + x_3 + x_1 \, x_3} \,,  
         x_1 + x_3 + x_1 \, x_3    \,, 
       \frac{x_2\,x_3}{x_3 + x_1\,(1 + x_2)\,(1+x_3)} \right)   \,, 
\end{equation}
which is totally positive map. 
Moreover, the preceding map is not an involution, but is reversible and invertible. 
Conjugating map \eqref{x3b} by the scaling $x\mapsto \lambda\,x$ and then 
taking $\lambda=0$, the map reduces to the $x$ component of the map of the previous case.
A final comment for the map \eqref{x3b} is that it is not covered in the classification
 \cite{sergeev1998}, and up to our knowledge is a new solution of the functional 
 tetrahedron equation.

\vspace{0.5cm}\noindent
{\em Case c:} Let us invoke the invariants of the subgroup generated by $S_1=\{X_1,X_3\}$, 
given by \eqref{invS1}, and the corresponding functionally independent relations \eqref{x2_set1}. 
In this case dpKP equation \eqref{dpKP} can be written in the invariant form
$$
\frac{1 - y_1 \,y_3}{1-y_1} = \frac{(1 - y_2)\, (1 + y_3)}{1-v_1} 
+ \frac{(1 - y_2)^2}{(1-v_1)^2}\, \frac{(v_2 - y_3)}{(1-v_2)} \,.
$$
The system of equations take a simpler form conjugating the invariants $(x_i,u_i)$
by the involution $(x_i,u_i)\mapsto (\frac{1}{x_i},\frac{1}{u_i})$ followed by a $Mob^{12}$
transformation $x\mapsto 1+x$ on all 12 invariants. The resulting system
of equations has two solutions.

The first tetrahedron map is 
\begin{equation}
\begin{array}{*3{>{\displaystyle}l}}
u_1 = \frac{x_1\,x_2\,(1+x_3)}{x_1 + x_3 + x_1\, x_3+x_2\,x_3}, &
u_2 = x_1 + x_3 + x_1\, x_3, &
u_3 = \frac{x_2\, x_3}{x_1 + x_3 + x_1\, x_3}, \\ [0.5cm]
v_1 = \frac{y_1 \,y_2}{y_1 + y_3 + y_1\, y_3}, &
v_2 = y_1 + y_3 + y_1\, y_3, &
v_3 = \frac{(1 + y_1)\, y_2 \, y_3}{y_1 + y_3 + y_1\,y_3 + y_1 \,y_2} \,.
\end{array} 
\end{equation}
The map decouples into two components $R^x$, $R^y$ which are simply related by 
$R^x=\tau_{13}\, R^y\, \tau_{13}$, and is a conjugation of tetrahedron map (25) in
Sergeev classification \cite{sergeev1998}.

The second solution leads to the following totally positive, coupled, tetrahedron map
\begin{flalign} \label{ftmaster}
& \left\lbrace
\begin{array}{*1{>{\displaystyle}l}} 
u_1 = \frac{ x_2\,y_1 \,(1+x_2)\,(1+y_3)^2}{ y_1\,(1+x_{2}) \,(1+y_3 ) + y_3} \,,  \\ [0.5cm]
v_1 = \frac{y_1\,y_2\,(1+x_2)\,(1+y_3)}{y_1\,(1+x_2)\,(1+y_3)+y_3} \,, 
\end{array} \right. \nonumber \\
& \left\lbrace
\begin{array}{*1{>{\displaystyle}l}}
u_2 = x_1+x_3+x_{1}\,x_{3} \,,  \\ [0.3cm]
v_2 = \frac{\big[ y_1\,\left(1+x_2\right)\,\left(1+y_3\right) + 
y_3\big]^{2}}{y_1\,\left(1+x_{2}\right)\,\left(1+y_{3}\right)\,\big[1+x_{2}\, \left(1+y_{3}\right) \big] + y_{3}} \,, 
\end{array} \right. \\
& \left\lbrace
\begin{array}{*1{>{\displaystyle}l}}
u_3 = \frac{x_2\,y_3}{y_1\,(1+x_2) \, \big[(1+x_2)\,(1+y_1)\,(1+y_3) + y_3 \big] + y_3}  \,, \\ [0.5cm]
v_3 = \frac{y_{2}\,y_{3}}{y_1 \, \left(1+x_2 \right) \,\left(1+y_2 \right) \,\left(1+y_3 \right) + y_3}  \,.
\end{array} \right. \nonumber
\end{flalign}
Substituting the symmetry constraint 
$$x_3\,y_1 \, (1+x_2)\, (1+y_3) = x_1\,y_3 \, (1+x_3)\,, $$ 
into the above relations and taking $y$ invariants to vanish we retrieve tetrahedron map
 \eqref{x3b}. 
Conjugating the map \eqref{ftmaster} by the scaling $s\mapsto \lambda\,s$ and then taking $\lambda=0$, 
the map reduces to the following tetrahedron map of triangular form
\begin{equation}
 R:(x_1,y_1,x_2,y_2,x_3,y_3) \mapsto
 \left( \frac{x_{2} y_{1}}{y_{1} + y_{3}}, \frac{y_{1} y_{2}}{y_{1} + y_{3}} , 
                            x_{1} + x_{3}, y_{1} + y_{3}, 
 \frac{x_{2} y_{3}}{y_{1} + y_{3}} ,  \frac{y_{2} y_{3}}{y_{1} + y_{3}}\right) \,,
\end{equation}
which has an obvious further reduction when $x$ and $y$ coincide component--wise. 
Thus tetrahedron map \eqref{ftmaster} can be regarded as a vector extension 
of tetrahedron map \eqref{x3b}, and degenerates to the simple tetrahedron map 
\eqref{octakp} in various ways. 

\subsection{Discrete modified KP ($\chi_4$) and $(\chi_5)$}
Equation $\chi_4$ in ABS classification, or the
discrete modified KP equation (Nijhoff and Capel) reads
\begin{equation}
 \frac{f_{13} - f_{12}}{f_1} + \frac{f_{12} - f_{23}}{f_2} + \frac{f_{23} - f_{13}}{f_3}=0 \,,
\end{equation}
and is invariant along the orbits of the abelian group generated by the symmetry vector
fields
$$ X_1= f\,\partial_f\,,\qquad X_2=(-1)^{n+m+k}f\, \partial_f\,.  $$  
We use the corresponding invariants \eqref{inv_ratios} since the symmetries of the two 
lattice equation are the same. Lattice equation $\chi_3$ is written in terms of the invariants 
in two different ways
\begin{equation}\label{modkp}
(1-x_1)\,(1-x_2)=(1-u_1)\,(1-u_2)\,, \, \qquad y_1\,(y_2-v_3) + y_3\,(v_1-y_2) + v_2\,(v_3 -v_1) =0 \,.   
\end{equation}
The corresponding tetrahedron map is obtained by solving the system of equations \eqref{akp_set1}, 
except the first one, along with any of the equations \eqref{modkp}, leading to the unique solution
\begin{equation} \label{ftx4}
\begin{array}{*3{>{\displaystyle}l}}
u_1 = \frac{Q}{1 - x_1\, x_3}\,, & 
u_2 = x_1 \, x_3\,, & 
u_3 = \frac{x_2\, (1 - x_1\, x_3)}{Q}\,, \\ [0.5cm]
v_1 = x_3\, y_2\,,   &
v_2 = \frac{Q\,y_3}{1 - x_1\, x_3} \,, &
v_3 = \frac{y_2}{x_1}\,, \end{array} 
\end{equation}
where $Q=x_2 + x_1\,(1 - x_2 - x_3)$. 
Under the conjugation $x\mapsto 1-x$, the map $R^x$ is gauge equivalent to tetrahedron map (25) 
in Sergeev classification. However, in contrast to its preceding appearance for the previous lattice 
equations now the vector map has a coupled triangular form, and it is an involution. 
  
Finally, equation $\chi_5$ in ABS classification, or the generalized lattice Toda equation 
(Nijhoff and Capel) is the lattice equation
\begin{equation} \label{eqx5}
 \frac{f_{13} - f_{23}}{f_3} = f_{12}\,\left(\frac{1}{f_2}-\frac{1}{f_1}\right) \,.
\end{equation}
It is a degeneration of equation $\chi_4$ by the limiting procedure
$f\mapsto \alpha^m\,\beta^n\,\gamma^k\,f$, and then taking $\gamma \rightarrow 0$, and 
$\alpha = \beta = 1$. 
The symmetries of the degenerated lattice equation \eqref{eqx5} are preserved, 
and $\chi_5$ is written in terms of the invariants
\eqref{inv_ratios} in the following forms
$$u_2 \, (1-u_1)=x_2\,(1-x_1) \,, \qquad y_1\,(v_3-y_2) + v_1\,(v_2-y_3)=0  \,.$$
The map obtained by solving the associated system of equations is
\begin{equation} \label{ftx5}
\begin{array}{*3{>{\displaystyle}l}}
u_1 = \frac{Q}{x_1\, x_3}\,, & 
u_2 = x_1 \, x_3\,, & 
u_3 = \frac{x_1\,x_2\, x_3}{Q}\,, \\ [0.5cm]
v_1 = x_3\, y_2\,,   &
v_2 = \frac{ Q \,y_3}{x_1\, x_3} \,, &
v_3 = \frac{y_2}{x_1}\,. \end{array} 
\end{equation}
where $Q= (x_1-1)\, x_2 + x_1 \, x_3$. 
The map (\ref{ftx5}) does not satisfy the tetrahedron relation on its own. Instead 
the following entwining relation  holds
\begin{equation}
  R_{\chi_4}^{(123)}\, R_{\chi_5}^{(145)}\, R_{\chi_5}^{(246)}\,R_{\chi_5}^{(356)} \, 
= R_{\chi_5}^{(356)}\, R_{\chi_5}^{(246)}\, R_{\chi_5}^{(145)}\,R_{\chi_4}^{(123)} , 
\label{eq:YBrel}
\end{equation}
where $R_{x_4}^{(123)}$, $R_{x_5}^{(246)}$ etc. denote the maps  (\ref{ftx4}),
(\ref{ftx5}), obtained by the lattice
equations ($\chi_4$), ($\chi_5$), respectively.

\section{Multi--dimensional and multi--field generalizations} \label{multi}
\subsection{Lattice potential KP on $\mathbb{Z}^4$} 
We take now the opportunity to derive the 4D consistency of a non--commutative version of dpKP 
that we are going to use later on.
We consider a hypercube in $\mathbb{Z}^4$ and copies of {\it dpKP} equation on each 
of its cubic faces that we label as follows:
\begin{equation} 
dpKP(ijk)
\equiv f_{ij}( f_i- f_j)+f_{jk}( f_j- f_k)+f_{ik}( f_k- f_i)=0 \, , \label{eq:KPijk}
\end{equation} 
and  

\begin{equation} dpKP(ijk)_l \equiv f_{ijl}( f_{il}- f_{jl})+f_{jkl}( f_{jl} - 
f_{kl})+f_{ikl}( f_{kl}- f_{il})=0 \,,
\label{eq:KPijkl}
\end{equation}
for their shifted versions in the remaining fourth direction respectively, where the 
indices $i,j,k,l$ satisfy $\left\{ i,j,k,l \right\}=\left\{ 1,2,3,4 \right\}, i<j<k.$ 
\footnote{Double and triple indices containing $l$ are reordered accordingly 
(i.e. for $l=1$, $f_{21}\equiv f_{12}$, $f_{241}\equiv f_{124}$ etc.)}
$dpKP$ equation is consistent on the 4D hypercube in the following sense:
\begin{figure}[h!]
\centering
  \def\svgscale{0.085}
  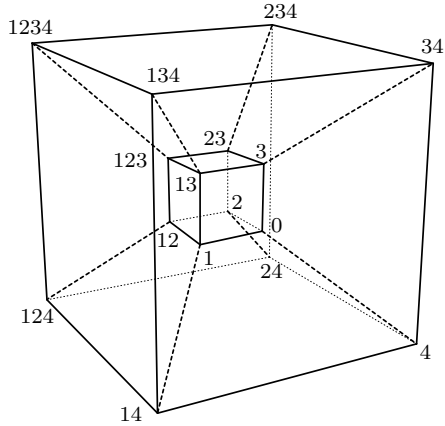
  \caption{dKP equation on $\mathbb{Z}^4.$ }
  \label{fig:tessera}
\end{figure}

Consider generic initial values $ f_1, f_2,  f_3, f_4,$  $ f_{13}, f_{23},
f_{14} $, $f_{123}, f_{134}$ on the nine vertices of the
hypercube (Fig. \ref{fig:tessera}) and the set of equations on all eight cubic faces 
partitioned in two subsets: 
$$ \mathbf {S}_{\mathbf {ILDB}}= \left\{ dpKP(123)=0, dpKP(234)_1=0, dpKP(124)=0, dpKP(134)_2=0  \right\}$$
and 
$$\mathbf {S}_{\mathbf {FURO}}=\left\{ dpKP(134)=0,  dpKP(124)_3=0, dpKP(234)=0, dpKP(123)_4=0  \right\}. $$
The first set contains the {\it `In, `Left', `Down', `Back'} cubes and the second the 
{\it `Front', `Up', `Right', `Out'} ones.
$ f_{12}, f_{24}$, $f_{124}, f_{234}$ can be calculated using $\mathbf {S}_{\mathbf {IFDR}}$ and 
$ f_{34}, f_{24}$, $f_{234}, f_{124}$ can be calculated using $\mathbf {S}_{\mathbf {LUBO}}$  
in terms of the 9 values given  initially. The two different ways of calculating the values of 
$f_{24}, f_{124}, f_{234}$ yield the same result. In particular we have:
\begin{subequations}
\begin{eqnarray}
f_{12}&=& \left[f_{13}(f_{1}-f_{3})+f_{23}(f_3-f_2)\right]( f_1 -f_{2})^{-1} \,, \label{f12}
\\
f_{34}&=& [f_{13}(f_{1}-f_{3})+f_{14}(f_4-f_1)](f_4-f_{3})^{-1}\,, \label{f34}
\\
f_{24}&=& [f_{13}(f_{1}-f_{3})+f_{23}(f_3-f_2)+f_{14}(f_4-f_1)](f_4 -f_{2})^{-1}\,, \label{f24}
\\
f_{234}&=&  [f_{134}(f_{14}-f_{13})(f_4-f_1)- f_{123}(f_{13}-f_{23}) (f_3-f_4)]\times \nonumber \\
& &\qquad \qquad  \qquad [(f_{13}-f_{14})f_1+ (f_{23} -f_{13})f_3 + (f_{14}-f_{23})f_4]^{-1} \label{f234}\,, \\
f_{124}&= & [ f_{134}(f_{14}-f_{13})(f_2-f_1)- f_{123}(f_{13}-f_{23}) (f_3-f_2)] \times \nonumber 
\\ 
& & \qquad \qquad \qquad  [f_{13}-f_{14})f_1(+ (f_{23} -f_{13})f_3+ (f_{14}-f_{23}) f_2]^{-1} \,.\label{f124}
\end{eqnarray}
\end{subequations}
This fact is in accordance to the $n$-dimensional consistency property of the underlying 
discrete KdV equation,
\begin{equation} \label{KdV}
(f_{12}-f)(f_1-f_2) = \alpha_1 -\alpha_2 \,.
\end{equation}
In fact, it was shown in \cite{abs2003} that equation \eqref{KdV} enjoys the
$n$-dimensional consistency property. What is even more interesting is that
solutions of the latter provide solutions of lattice dpKP equation. Indeed, consider
three copies of discrete KdV \eqref{KdV} on the faces of a cube
\begin{eqnarray}
	(f_{12}-f)(f_1-f_2)- \alpha_1+\alpha_2&=& 0\,, \\
	(f_{23}-f)(f_2-f_3)- \alpha_2+\alpha_3&=& 0\,, \\
  (f_{13}-f)(f_3-f_1)- \alpha_3+\alpha_1&=& 0\,.
\end{eqnarray}
By adding them up we obtain readily (\ref{dpKP}).
This means that solutions of the compatible dpKdV equation on $\mathbb{Z}^3$ provide solutions 
to the  dpKP equation.

One of the main findings in \cite{ptv2006} was that YB maps may 
be also obtained from quad--graph integrable lattice equations, 
if one considers the extension of lattice equations on a multidimensional lattice along with the 
invariants of the admitted multi--parameter symmetry groups. Thus, it would be interesting to 
investigate whether this property is inherited to genuine three dimensional lattice equations, 
such as the discrete potential KP equation.

Let us  assign the following 3--fold joint invariants along the orbits of the full symmetry group 
of discrete potential KP equation:
\begin{equation}  \label{4dpkp}
  x = \frac{f_1 - f_3}{f_3 - f_2} \,,\quad 
  v = \frac{f_{1} - f_{4}}{f_{4} - f_{2}} \,,\quad 
  z = \frac{f_{24} - f_{23}}{f_{23} - f_{34}} \,, \quad
  u = x_4 \, \quad
  y = v_3 \,, \quad
  w = z_1 \,, \quad
 \end{equation}  
They naturally live on the two--dimensional faces of the 3--cubes
labeled by ``In'', ``Down'', ``Right'' and their shifts in the 4, 3, 1
directions of the $\mathbb{Z}^4$ lattice, respectively.
We note that the invariants \eqref{4dpkp} are straightforward multi--dimensional 
generalization of the joint invariants used in \cite{ptv2006} in connection with 
the lattice KdV equation on $\mathbb{Z}^4$. 
In the present case, the full symmetry group of the dKdV corresponds
to the subgroup generated by $\{X_1,X_2,X_4\}$,
isomorphic to the group of isometries of Minkowski plane. 
  
In order to derive the functional relations among the invariants \eqref{4dpkp}, 
we use the symmetric form of dpKP \eqref{dpKP}. Accordingly, the set of 
lattice dpKP imposed on all 3--cubes of a hypercube can be written as
\begin{equation}
A_{ijk} = A_{ikj} = A_{jik} \,, \qquad A_{ijk} := \frac{f_{ij}-f_{jk}}{f_i-f_k} \,,
\end{equation}
adopting the preceding convention on the indices. From the defining relations of the 
invariants and the KP equations on the ``Front'' and ``Right'', cubes we get that the
invariants are related by
\begin{equation} \label{tes1a}
x\,y = u\,v \\, \qquad v = x+z+x\,z\,.
\end{equation}
On the other hand, from the KP on the ``Out''  and  ``Up''  cubes and
the defining relations of the invariants it follows that
\begin{equation} \label{tes1b}
 y = u+w+u\,w\,.
\end{equation}
Solving equations \eqref{tes1a}, \eqref{tes1b} for $u,v,w$ we get the map
\begin{equation}
 u = \frac{x\,y}{x+z+x\,z} \,, \quad v = x+z+x\,z  \,,\quad  
 w = \frac{y\,z\,(1+x)}{x+z+x\,y+x\,z}\,,
\end{equation}
which is the tetrahedron map (25) in Sergeev \cite{sergeev1998} classification. 
Thus, discrete potential KP, and its intrinsic multi--dimensional consistency property, 
encapsulates many different known and new examples of tetrahedron maps.

\subsection{Tetrahedron map from non--commutative discrete potential KP }

Next we show how our approach can be applied to simple non--commutative versions
of the previous lattice equations to obtain non--commutative tetrahedron maps. 
We focus on non--commutative dpKP equation and in particular its symmetry subgroup 
considered in Section \ref{kpx3} ({\em case b})

\begin{equation} \label{ncKP}
 F_{12} \, (F_1 - F_2) + F_{23}\, (F_2 - F_3) + F_{13} (F_3 - F_1)=0 \,.
\end{equation}
Let $\sigma$ be the cycle matrix 
$$\sigma = \left[\begin{array}{ccccc}
0  & 1  &        &        &   \\
   & 0  & 1      &        &   \\
   &    & \ddots & \ddots &   \\
   &    &        & \ddots & 1 \\
1  &    &        &        & 0
\end{array}\right] \,,$$
and $F$ a generalized cycle matrix i.e.
$$ F = 
\,\left[\begin{array}{ccccc}
0        & f^{(2)}   &         &        &         \\
         & 0         & f^{(3)} &        &         \\
         &           & \ddots  & \ddots &         \\
         &           &         & \ddots & f^{(n)} \\
f^{(1)}  &           &         &        & 0
\end{array}\right] =  \sigma \, \widetilde{F} \,,
$$
where $\widetilde{F} \equiv {\rm diag} (f^{(1)} , f^{(2)} , \cdots , f^{(n)} )$. Then 
\eqref{ncKP} can be written as a generalized cycle matrix equation. 
The corresponding matrix invariants are
\begin{equation} \label{inv_ratiosnc}
x_1 = {f_1}\,{f_2}^{-1},\, x_2 = f_{12}\,{f_{23}}^{-1},\, 
x_3= {f_2}\,{f_3}^{-1},\, 
u_1 = {f_{13}}\,{f_{23}}^{-1},\, u_2 = f_{1}\,{f_{3}}^{-1}, \, 
u_3 = {f_{12}}\,{f_{13}}^{-1}\,.
\end{equation}
Notice that the invariants are diagonal matrices. From \eqref{inv_ratiosnc} 
we immediately find that they are functionally related by
\begin{equation} \label{ncinvrel}
 x_1\,x_3 = u_2 \,, \quad  u_1\,u_3=x_2 \,.
\end{equation} 
On the other hand matrix dpKP \eqref{ncKP} is written in terms of the invariants 
\eqref{inv_ratiosnc} in the following form 
\begin{equation} \label{nckpinv}
 (u_1-{\rm Id})\, \sigma \, (u_2-{\rm Id}) \, {u_2}^{-1} = 
   (x_2-{\rm Id})\, \sigma \, (x_1-{\rm Id}) \, {x_1}^{-1} \,.   
\end{equation}   
The system of equations \eqref{ncinvrel}, \eqref{nckpinv} can be easily solved uniquely in terms
for $u_i$ (and vice versa for $x_i$). The fact that the resulting invertible map satisfies the
tetrahedron property follows from the four dimensional consistency property of 
non--commutative dpKP equation \eqref{dpKP} proved in Section \ref{multi}. 
In the scalar commutative case relations \eqref{ncinvrel} and \eqref{nckpinv} reduce 
to the relations defining tetrahedron map \eqref{x3b}, under the conjugation $x\mapsto 1+x$.

\subsection{Discrete Calapso KP}
Let us consider the following vector generalisation of the
dpKdV  
\begin{equation} \label{cal}
(\,\boldsymbol{f}_{12}-\boldsymbol{f} \,) =
\frac{\alpha_1-\alpha_2}{|\boldsymbol{f}_{1}- 
\boldsymbol{f}_{2}|^2} \, (\,\boldsymbol{f}_{1}-\boldsymbol{f}_2 \,)\,,
\end{equation}
$\boldsymbol{f}:\mathbb{Z}^2 \rightarrow \mathbb{C}^n$, introduced by Schief in
\cite{calapso}, under the name {\em discrete Calapso equation}. Imposing
\eqref{cal} on $\mathbb{Z}^3$, and using its three dimensional consistency
property (see Eq. (78) in \cite{ptv2006} ), one finds that solutions of Eq. 
\eqref{cal} provide solutions of the following vector 3D lattice equation
\begin{equation} \label{dcKP}
  \frac{|\boldsymbol{f}_2-\boldsymbol{f}_3|^2}{|\boldsymbol{f}_{12}-\boldsymbol{f}_{13}|^2}\,(\boldsymbol{f}_{12}-\boldsymbol{f}_{13}) +
  \frac{|\boldsymbol{f}_1-\boldsymbol{f}_2|^2}{|\boldsymbol{f}_{13}-\boldsymbol{f}_{23}|^2}\,(\boldsymbol{f}_{13}-\boldsymbol{f}_{23})+
  \frac{|\boldsymbol{f}_1-\boldsymbol{f}_3|^2}{|\boldsymbol{f}_{23}-\boldsymbol{f}_{12}|^2}\,(\boldsymbol{f}_{23}-\boldsymbol{f}_{12})=0\,.
\end{equation}
The latter equation may be considered as a discrete 3D discrete Calapso equation. In the scalar
(complex) case lattice equation \eqref{dcKP} is a correspondence, i.e. it is
satisfied on solutions of the original KP equation, as well as on solutions of
the following complexification
\begin{equation} \label{comkp}
  f_{12} (\bar{f}_1-\bar{f}_2) + f_{23} (\bar{f}_2-\bar{f}_3) + f_{13} (\bar{f}_3-\bar{f}_1)=0\,.
\end{equation}
Using the invariants of the translational symmetry of \eqref{comkp}, we obtain
the following tetrahedron map
\begin{equation} 
 R_c(x,y,z) = \left( \frac{\bar{x}\,y}{\bar{x}+\bar{z}} , x+z,\frac{\bar{z}\, y}{\bar{x}+\bar{z}}\right)\,.
\end{equation}
The above tetrahedron map is closely related with tetrahedron map \eqref{octakp}, which for 
convenience we rewrite below 
\begin{equation} \label{octakp2}
 R(x,y,z) = \left( \frac{x\,y}{x+z} , x+z,\frac{z\, y}{x+z}\right)\,.
\end{equation}
 Indeed, the
antiholomorphic involution $\sigma:x\mapsto \bar{x}$ is a symmetry of the FT
map \eqref{octakp2}, thus according to Proposition \ref{prop:symm2} the map
$$ \widetilde{R} = \sigma \times Id \times \sigma \, R  \, Id \times
\sigma \times Id \,, $$
is a new solution of tetrahedron equation which coincides with the $R_c$ map above, since 
\begin{equation}
(x,y,z) \overset{\sigma_2}{\longrightarrow} ({x},\bar{y},{z}) \overset{R}{\longrightarrow}
\left( \frac{{x}\,\bar{y}}{{x}+{z}} , {x}+{z},\frac{{z}\,
    \bar{y}}{{x}+{z}}\right) \overset{\sigma_{13}}{\longrightarrow} \left(
  \frac{\bar{x}\,y}{\bar{x}+\bar{z}} , x+z,\frac{\bar{z}\,
    y}{\bar{x}+\bar{z}}\right) \,,
\end{equation}
where $\sigma_2 = {\mathrm{Id}} \times \sigma \times {\mathrm{Id}} $ and 
$\sigma_{13} = \sigma \times {\mathrm{Id}} \times \sigma$.

Next, we investigate the four dimensional case by using the well known correspondence (quaternions) 
  $$ \boldsymbol{v} \in \mathbb{R}^4 \qquad \longleftrightarrow \qquad V = 
  \begin{pmatrix}   \phantom{-}v_1 + {\rm{i}} \, v_2  &  v_3 + {\rm{i}} \, v_4 \\
                              -v_3 + {\rm{i}} \, v_4  &  v_1 - {\rm{i}} \, v_2  \end{pmatrix}\,, \qquad
V^{\,-1} = \frac{V^{\,\ast}}{\det V} \,,$$
where asterisk denotes Hermitian conjugation. 
 In terms of the $G$-invariants 
 $$\boldsymbol{x}=\boldsymbol{f}_1-\boldsymbol{f}_2\,,\quad \boldsymbol{y}=\boldsymbol{f}_{12}-\boldsymbol{f}_{23}\,,\quad
 \boldsymbol{z}=\boldsymbol{f}_2-\boldsymbol{f}_3\,,\quad
 \boldsymbol{u}=\boldsymbol{f}_{13}-\boldsymbol{f}_{23}\,,\quad \boldsymbol{v}=\boldsymbol{f}_1-\boldsymbol{f}_3\,,\quad
 \boldsymbol{w}=\boldsymbol{f}_{12}-\boldsymbol{f}_{13}\,, $$
of the translational invariance of \eqref{dcKP} we obtain the following matrix system 
\begin{equation} \label{quatsys}
Y=U+W\,,\qquad V=X+Z \,, \qquad \frac{\det Z}{\det W}\, W+\frac{\det X}{\det U}\, U = \frac{\det V}{\det Y}\, Y \,.
\end{equation}  
The system \eqref{quatsys} admits the solution 
\begin{equation} \label{ncft} 
 U = Y\,X\,(X+Z)^{-1} \,,\quad   V = X+Z\,, \quad W = Y\,Z\,(X+Z)^{-1}\,, 
 \end{equation}
as one can verify by straightforward calculations using the relation $W\,X=U\,Z$. By 
Hermitian conjugation of the preceding solution we obtain that
$$  U=(X+Z)^{-1}\,X\,Y \,, \quad V = X+Z\,, \quad W = (X+Z)^{-1}\,Z\,Y\,, $$
is also a solution of system \eqref{quatsys}. Both solutions satisfy the tetrahedron 
equation, and they are instances of non--commutative versions of map \eqref{octakp2}.

Similarly, two more solutions of \eqref{quatsys} read
$$
\begin{array}{lcl}
  U = Y\,X^\ast\,(X^\ast+Z^\ast)^{-1} \,, & V= X+Z\,, 
  & W = Y\,Z^\ast\,(X^\ast+Z^\ast)^{-1}\,, \\  [0.2cm]
  U = (X^\ast+Z^\ast)^{-1}\,X^\ast\,Y \,, & V= X+Z\,, 
  & W = (X^\ast+Z^\ast)^{-1}\,Z^\ast\,Y \,, 
\end{array}
$$
which do not satisfy the tetrahedron property. 
However, in the case of $\mathbb{R}^3$ (pure quaternions) all solutions are gauge equivalent 
to \eqref{ncft}, since in this case all matrices are skew--Hermitian, i.e. $X+X^\ast=0$ etc.

\section{Conclusions and Perspectives} \label{concl}

In this paper we presented a generalization of the method developed in \cite{ptv2006}, 
to investigate a relationship between tetrahedron maps and the consistency property of 
integrable discrete equations on $\mathbb{Z}^3$. 
The method is demonstrated by a case--by--case study of the octahedron type lattice 
equations classified recently by Adler, Bobenko and Suris, as well as BKP lattice equation and their 
non--commutative versions  \cite{nc1990} \cite{doko}, \cite{nimmo}, leading to some new 
examples of tetrahedron maps. 
A question which arises naturally is how this link can be explored when the lattice equations 
does not possess enough, if any, symmetry. 

Let us address briefly this issue here by considering a concrete example.  
Invoking discrete modified KP equation ($\chi_4$) it can be written equivalently in the 
following form 
\begin{equation}
   \frac{f\, f_{13}}{f_1\,f_3} \, (f_3-f_1) + 
   \frac{f\, f_{12}}{f_1\,f_2} \, (f_1-f_2) + 
   \frac{f\, f_{23}}{f_2\,f_3} \, (f_2-f_3) =0 \,.
\end{equation} 
Clearly the equation does not admit any translational symmetry. However, if we insist in  
writing it in terms of the invariants of the translational symmetry (the $x,u$ invariants in
\eqref{translation}) then we are forced to take into account also the semi invariants of its
symmetry group, which are the same as for AKP and BKP equations given by \eqref{semi}. 
Denoting the latter invariants by $y_i,v_i$, discrete modified KP 
can be 
written in the following invariant form
$$ u_2 \, v_2 = x_1\,y_1+x_3\,y_3\,. $$
On the other hand the invariants are related by 
$$ u_1+u_3 = x_2 \,, \quad  x_1+x_3=u_2 \,,\quad 
y_1\, y_2 = v_1\,v_2\,,\qquad y_2\,y_3=v_2\,v_3 \,. $$
Since not both $u_1,u_3$ are determined by the above system, and we know that the map \eqref{octakp} 
satisfies the tetrahedron equation, let us augment the above five eqautions with the relation 
$$ u_1\,u_2 = x_1 \, x_2\,.$$
Then the unique solution of the above system of six equations reads
\begin{equation} \label{x4symm}
\begin{array}{*3{>{\displaystyle}l}}
u_1 = \frac{x_1\,x_2}{x_1 + x_3}\,, & 
u_2 = x_1 + x_3\,, & 
u_3 = \frac{x_2\,x_3}{x_1 + x_3} \,, \\ [0.5cm]
v_1 = \frac{(x_1 + x_3)\,y_1\,y_2}{x_1\,y_1+x_3\,y_3} \,,   &
v_2 = \frac{x_1\,y_1 + x_3\,y_3}{x_1+x_3} \,, &
v_3 = \frac{(x_1 + x_3)\,y_2\,y_3}{x_1\,y_1+x_3\,y_3}\,, \end{array} 
\end{equation}
in terms of elementary multi--symmetric polynomials on two sets of three variables. 
The above map satisfies the tetrahedron relation and is
both an involution and a reversible map. Moreover, the inverse map is obtained 
by just interchanging the variables $u_i\leftrightarrow x_i$,
$y_i\leftrightarrow v_i$,  $i=1,2,3$. 
Introducing inhomogeneous (or ``projective") variables $p_i=x_i\,y_i$, and $q_i=u_i\,v_i$,
$i=1,2,3$, the resulting map readily reduces to the tetrahedron map
\eqref{octakp}. 
These observations suggest that discrete modified KP equation can be coupled, 
in a consistent way, with some other lattice equation for a new potential, say $h$, 
which is invariant under translations. 
Indeed, the relation $u_1\,u_2=x_1\,x_2$, that was added, is nothing but the conserved form 
of discrete potential KP ($\chi_3$) equation \eqref{KP_symm}, and the corresponding 
multi--dimensional consistent, coupled 3D lattice system is 
\begin{eqnarray}
  {f_2\, f_{13}} \, (h_3-h_1) + {f_3\, f_{12}} \, (h_1-h_2) + {f_1\, f_{23}} \, (h_2-h_3) &=&0 \,, \\ 
h_{13} \, (h_3 - h_1) + h_{12}\, (h_1 - h_2) + h_{23}\, (h_2 - h_3)   &=&0 \,.
\end{eqnarray}
Apparently, the potential $h$ is invariant under translations, so one speaks for a hidden 
{\em potential symmetry group} of the starting lattice equation ($\chi_4$). 
We postpone to give a detailed analysis on these issues in a future publication, 
and refer to \cite{knpt2019} for a recent account along this line of research for integrable 
discrete equations on quad--graphs.

\small
\bibliographystyle{plain}

\Addresses

\end{document}